\documentclass[manuscript, nonacm]{acmart}

\usepackage[ruled,noline,linesnumbered, noend]{algorithm2e}

\newcommand{\algotab}{\hspace{12pt}}
\let\oldnl\nl
\newcommand{\nonl}{\renewcommand{\nl}{\let\nl\oldnl}}

\newtheorem{theorem}{Theorem}[section]
\newtheorem{corollary}[theorem]{Corollary}
\newtheorem{lemma}[theorem]{Lemma}

\newtheorem{definition}[theorem]{Definition}
\newtheorem{remark}[theorem]{Remark}

\newtheorem{assumption}{Assumption}
\AtBeginDocument{%
  \providecommand\BibTeX{{%
    \normalfont B\kern-0.5em{\scshape i\kern-0.25em b}\kern-0.8em\TeX}}}

\acmConference[PODC '21]{PODC '21: ACM Symposium on Principles of Distributed Computing}{July 26--30, 2021}{Virtual}
\begin{document}

\title{Who Needs Consensus? A Distributed Monetary System Between Rational Agents via Hearsay}

\author{Yanni Georghiades}
\email{yanni.georghiades@utexas.edu}
\affiliation{%
  \institution{The University of Texas at Austin}
  \city{Austin}         
  \country{United States}   
}

\author{Robert Streit}
\email{rpstreit@utexas.edu}
\affiliation{%
  \institution{The University of Texas at Austin}
  \city{Austin}         
  \country{United States}   
}

\author{Vijay Garg}
\email{garg@utexas.edu}
\affiliation{%
  \institution{The University of Texas at Austin}
  \city{Austin}         
  \country{United States}   
}


\begin{abstract}
  We propose a novel distributed monetary system called Hearsay that tolerates both Byzantine and rational behavior without the need for agents to reach consensus on executed transactions.
  Recent work \cite{guerraoui2019consensus, collins2020online, auvolat2020money} has shown that distributed monetary systems do not require consensus and can operate using a broadcast primitive with weaker guarantees, such as reliable broadcast. 
  However, these protocols assume that some number of agents may be Byzantine and the remaining agents are perfectly correct.
  For the application of a monetary system in which the agents are real people with economic interests, the assumption that agents are perfectly correct may be too strong. 
  We expand upon this line of thought by weakening the assumption of correctness and instead adopting a fault tolerance model which allows up to $t < \frac{N}{3}$ agents to be Byzantine and the remaining agents to be \textit{rational}.
  A rational agent is one which will deviate from the protocol if it is in their own best interest.
  Under this fault tolerance model, Hearsay implements a monetary system in which all rational agents achieve agreement on executed transactions.
  Moreover, Hearsay requires only a single broadcast per transaction.
  In order to incentivize rational agents to behave correctly in Hearsay, agents are rewarded with transaction fees for participation in the protocol and punished for noticeable deviations from the protocol.
  Additionally, Hearsay uses a novel broadcast primitive called Rational Reliable Broadcast to ensure that agents can broadcast messages under Hearsay's fault tolerance model.
  Rational Reliable Broadcast achieves equivalent guarantees to Byzantine Reliable Broadcast \cite{bracha1987asynchronous} but can tolerate the presence of rational agents.
  To show this, we prove that following the Rational Reliable Broadcast protocol constitutes a Nash equilibrium between rational agents.
  We deem Rational Reliable Broadcast to be a secondary contribution of this work which may be of independent interest.
\end{abstract}



\maketitle

\section{Introduction}

Distributed monetary systems have risen in popularity in recent years in large part due to the success of Bitcoin \cite{Nakamoto2008} and the flood of distributed money transfer protocols which followed it. 
Broadly, these protocols can be categorized into three groups: protocols which require agents to solve Nakamoto consensus \cite{Nakamoto2008}, protocols which require agents to solve Byzantine fault tolerant consensus (which we often refer to as ``classical consensus'') \cite{fischer1985impossibility}, and protocols which require weaker notions of agreement than consensus such as reliable broadcast \cite{bracha1987asynchronous}. 

The most popular category of distributed money transfer protocols consists of work which aims to improve on Nakamoto consensus. 
While great headway has been made, Nakamoto consensus still has several drawbacks.
The most notorious ones are its limited transaction throughput and high transaction settlement times, which are in many ways a byproduct of its lack of finality in executed transactions.
Additionally, to the best of our knowledge, there have been no variants of Nakamoto consensus proposed which do not assume that some amount of decision-making power (typically hashing power for Proof-of-Work \cite{dwork1992pricing}) in the system is correct.  

The next most prominent category of distributed money transfer protocols includes protocols which use more classical Byzantine fault tolerant consensus protocols such as PBFT to agree upon executed transactions. 
These protocols typically offer transaction finality, but it comes at a cost. 
In general, classical consensus is a notoriously difficult problem to solve in asynchronous systems \cite{fischer1985impossibility} and is expensive in terms of both time and message complexity. 
Additionally, whereas Nakamoto consensus requires only a strict majority of correct decision-making power, classical consensus requires that greater than $\frac{2}{3}$ of the agents in the system are correct. 

The third--and most relevant--category of distributed money transfer protocols is composed of protocols which do not require any form of consensus, instead opting for a less costly form of agreement. 
In particular, a line of recent work \cite{guerraoui2019consensus, collins2020online, auvolat2020money} in distributed money transfer protocols which rely weaker broadcast primitives is a key source of inspiration for Hearsay.
The work of Guerraoui et. al. observes that the consensus number \cite{herlihy1991wait} of a distributed money transfer protocol is 1 and proposes a money transfer protocol which relies on secure reliable multicast \cite{malkhi1997high} in lieu of consensus. 
Following this work, Collins et. al. propose and implement a protocol called Astro which uses only a single reliable broadcast \cite{bracha1987asynchronous}  per transaction. 
Similarly, Auvolat et. al. presents a simplified and generic money transfer protocol which also uses only a single reliable broadcast per transaction. 
Following a slightly different approach, ABC \cite{sliwinski2019abc} is a Proof-of-Stake protocol which also requires only a weak form of consensus which does not require termination for messages with Byzantine senders. 
Notably, each of these protocols assumes the same fault tolerance model in which more than $\frac{2}{3}$ of agents must be correct.
This assumption is quite reasonable for many applications in distributed computing, but it may be too strong for a money transfer system in which agents are inherently economically motivated. 
It seems optimistic to assume that a substantial number of agents would remain correct if deviating from the protocol would provide greater money or lower costs. 
In order to counteract this type of self-serving behavior, we forego the assumption of correctness and instead assume that at least $\frac{2}{3}$ of agents are \emph{rational}. 
We extend the works of \cite{guerraoui2019consensus, collins2020online, auvolat2020money} by instilling mechanisms into the money transfer protocol which cause rational agents to prefer following the protocol over deviating. 

In this work we present a novel distributed monetary system called Hearsay that tolerates both Byzantine and rational behavior and operates without the need for consensus. 
Specifically, we contribute:
\begin{itemize}
    \item The Hearsay protocol, a distributed monetary system that tolerates up to $t < \frac{N}{3}$ Byzantine agents and assumes the remaining agents to be rational.
    \item The Rational Reliable Broadcast protocol, a variant of reliable broadcast that that also tolerates up to $t < \frac{N}{3}$ Byzantine agents and assumes the remaining agents to be rational.
    \item Formal analysis and guarantees of Hearsay and Rational Reliable Broadcast in a game theoretic environment.
\end{itemize}
In Section \ref{sec:related_work}, we summarize the related work that influenced our research. We overview distributed money transfer systems which use Nakamoto consensus and Byzantine fault tolerant consensus, distributed money transfer systems which require weaker guarantees than consensus, and 
other examples of applications of game theory in distributed computation. 
In Section \ref{sec:model}, we describe our model which assumes that a bounded number of agents are Byzantine while the remainder are rational. 
In Section \ref{sec:rrb}, we present Rational Reliable Broadcast (RRB), a secondary contribution of this work. 
RRB is a variant of the reliable broadcast protocol which operates in game theoretic environments. 
Not only do we show correctness in the presence of Byzantine faults, but also that our incentive mechanisms cause rational agents to behave correctly. 
Specifically, we show that following the RRB protocol is a Nash equilibrium. 
Rational Reliable Broadcast is the backbone of transaction broadcast in Hearsay, but we also believe that a reliable broadcast protocol with safety and liveness guarantees in the presence of rational agents is of independent interest. 
In Section \ref{sec:hearsay} we describe Hearsay, the main contribution of this work. 
We provide pseudo-code for the Hearsay protocol and prove that it satisfies \emph{Validity}, \emph{Agreement}, and \emph{Integrity} as defined in Section \ref{sec:hearsay_properties}. 
Hearsay achieves these properties in the presence of rational agents by imposing rewards for participation and punishments for deviation from the protocol through a transaction fee system. 
Section \ref{sec:incentives} examines this dynamic in depth, proving that correct behavior in the Hearsay protocol is a Nash equilibrium, as no rational agent can do better (gain more utility) by deviating from the protocol specification. 
Finally in Section \ref{sec:future} we reflect and discuss avenues for future work.

\section{Related Work}\label{sec:related_work}

This paper extends a line of recent work in distributed monetary systems \cite{guerraoui2019consensus, collins2020online, auvolat2020money} which do not require consensus on transactions.
Instead, \cite{guerraoui2019consensus} uses secure reliable multicast \cite{malkhi1997high} and both \cite{collins2020online} and \cite{auvolat2020money} use reliable broadcast \cite{bracha1987asynchronous} for broadcasting messages.
Both broadcast protocols offer weaker guarantees than consensus, including a guarantee of source order rather than total order on message deliveries and a lack of any delivery guarantee for messages from Byzantine senders.
However, each of these papers makes use of the classical fault tolerance model in which up to $t < \frac{N}{3}$ of agents may be Byzantine and the rest are assumed to be correct.
In a monetary system in which the agents are entities with economic interests, this assumption of correctness may not hold in practice. 
Hearsay does not depend on consensus and assumes a weakened fault tolerance model inspired by \cite{aiyer2005bar}. 
In Hearsay, up to $t < \frac{N}{3}$ of agents may be Byzantine and the rest are assumed to be \textit{rational}, meaning they will deviate from the protocol if it is in their own best interest. 
In order to accommodate for this weaker fault tolerance model, we add mechanisms to both the Byzantine reliable broadcast protocol from \cite{bracha1987asynchronous} and to the monetary system from \cite{auvolat2020money} which incentivize rational agents to behave as correct agents. 

There have been numerous distributed money transfer protocols proposed recently, and the protocols which require a form of consensus on executed transactions can be broadly categorized into two groups.
The first group of protocols \cite{lewenberg2015inclusive,sompolinsky2015secure, eyal2016bitcoin} consists of those that adopt a model similar to that of Bitcoin. 
Bitcoin uses a blockchain to achieve Nakamoto consensus in asynchronous networks under the stated assumption that less than $\frac{1}{2}$ of the decision-making power (in this case, hashing power to be used for Proof-of-Work) is controlled by Byzantine agents and the remaining power is controlled by correct agents \cite{Pass2016}.
Later work shows that following the Bitcoin protocol is not incentive-compatible if all agents are assumed to be rational \cite{eyal2014majority, carlsten}. 
Another approach to improving blockchain-based money transfer protocols is to replace Nakamoto consensus with a more classical Byzantine fault tolerant consensus protocol such as PBFT \cite{castro1999practical}. 
These protocols \cite{abraham2016solida, pass2017hybrid} typically offer more classical safety and liveness guarantees than Nakamoto consensus, such as transaction finality.  
However, they still adopt the fault tolerance model that up to $t < \frac{N}{3}$ of agents may be Byzantine and that the rest are correct.
In an approach slightly more similar to that of Hearsay, ABC \cite{sliwinski2019abc} implements a distributed money transfer protocol which requires only a relaxed form of consensus on transactions. 
ABC is able to achieve transaction finality under the familiar assumption that up to $t < \frac{N}{3}$ of the decision-making power (in this case, asset ownership to be used for Proof-of-Stake) is controlled by Byzantine agents and the remaining power is controlled by correct agents.
While Hearsay is a permissioned system and therefore operates under a different execution model from the permissionless protocols, we hold these protocols up as points of contrast in order to highlight the novelty of Hearsay. 
Hearsay does not require consensus--either Nakamoto or classical--on transactions, but Hearsay's most significant deviation from previous distributed monetary systems is in its fault tolerance model, which foregoes any assumption that a majority of agents will behave correctly and instead assumes that these agents will behave rationally. 

There have been a few papers in recent years which apply rationality assumptions to distributed computing protocols. 
Despite the stated fault tolerance model used by Bitcoin, some argue that the Bitcoin protocol is indeed incentive compatible if all agents are rational \cite{biais2019blockchain}, and the effectiveness of transaction fees and block rewards in the Bitcoin implementation supports this claim to some degree. 
More explicitly, \cite{aiyer2005bar} introduces the BAR fault tolerance model, which captures three types of behavior: Byzantine, altruistic (i.e., correct), and rational. 
The same work introduces BAR-B, an asynchronous replicated state machine which achieves safety and liveness guarantees under the BAR model. 
The local retaliation policy used in Rational Reliable Broadcast is inspired by a similar policy used in BAR-B, and the fault tolerance model used in Hearsay is identical to the BAR model except that in Hearsay there are no altruistic agents. 
In the application of multiparty computation, \cite{halpern2004rational} shows that multiparty computation is incentive compatible under the assumption that all agents are rational. 
Extending this work, \cite{abraham2006distributed} shows the same result even if some agents behave unexpectedly and some agents collude in an attempt to increase their utility.

\section{Model}\label{sec:model}
We consider a system of N agents which store their own outgoing transaction log and a local view of the balance of each other agent. 
Agents communicate via direct, asynchronous, and reliable channels with bounded transmission times. We assume that channels are authenticated, so messages sent can be attributed to their initiator and the integrity of messages can be confirmed.
The system is \textit{permissioned}, and all agents are known to each other before communication begins. 
We allow up to $t < \frac{N}{3}$ agents to be \textit{Byzantine}, meaning they may exhibit arbitrary behavior.
This includes but is not limited to faulty behavior, collusion, and malice against other participants.

The remaining agents are assumed to be \textit{rational}, meaning that they make decisions which maximize their own utility gain or rewards.
Rational agents follow the protocol correctly if it is in their best interest to do so, but will deviate if they perceive a better outcome as a result.
Such an assumption is highly applicable in financial systems, such as the monetary system we consider in this work, because the application connects directly with ideas of loss and gain.
This means that we cannot simply show that our protocols are tolerant to Byzantine behavior, but we must also show that rational agents are incentivized to behave correctly.
In our results we show that correct behavior is a \emph{Nash equilibrium} amongst the rational agents. 
This means that under the assumption that all other rational agents are following the protocol, no agent can gain more utility by deviating from the protocol than by behaving correctly.
This argument shows that all rational agents behave correctly as doing otherwise could only reduce their expected utility gain.
In proving these game theoretic results, we assume that there are no externalities available to rational agents, such as extra communication, collusion, or non-rational behaviors such as spite (naturally Byzantine agents may still exhibit arbitrary behavior including these externalities). 
Examining our results in systems which allow collusion between rational agents is an interesting avenue of future work.
 
\section{Rational Reliable Broadcast}
\label{sec:rrb}
Agents in the Hearsay protocol uses reliable broadcast \cite{bracha1987asynchronous} in order to communicate messages.
Since we assume the non-Byzantine agents are rational, they may not execute the reliable broadcast protocol correctly if they see a benefit in deviating.
For example, if there is a cost associated with forwarding a message and an agent believes their participation is not instrumental (i.e. the rest of the network will ``take care of it''), then they will not forward messages. 
In the social sciences this is referred to as the \emph{free rider problem}, the idea that rational agents will choose to reap the benefits of a service without contributing to it.
Due to the large number of rational agents in our model, we must be careful in designing a suitable reliable broadcast protocol that overcomes these challenges. 
The following section details a novel broadcast protocol called Rational Reliable Broadcast (RRB) that we use to broadcast transactions in Hearsay.
\par We start with the foundational Byzantine reliable broadcast protocol (BRB) from \cite{bracha1987asynchronous}. BRB assumes that up to $t < \frac{N}{3}$ agents may be Byzantine and the remaining agents are correct.
The fault tolerance model used by Hearsay assumes that up to $t < \frac{N}{3}$ agents may be Byzantine but that the remaining agents are rational instead of correct.
Agents in RRB follow a simple retaliation policy in order to guarantee fault tolerance under our model. 

Each agent in Hearsay runs $N$ separate RRB instances--one for each agent--and an agent's $j^{th}$ RRB instance is for participating in broadcasts initiated by agent $p_j$.
We use this understanding in the specification of our reliable broadcast protocol and in our proofs. 
An equivalent view of this specification is that each agent has a separate broadcast channel, where the $i^{th}$ channel is owned by agent $p_i$ and $p_i$ is the only agent allowed to initiate broadcasts in this channel. 
Then when an agent executes any part of the RRB protocol on a particular instance, it accesses only the sequence numbers, queues, and variables pertinent to that instance, and retaliation can be performed within a single instance without impacting any other instance.
A RRB instance where agent $p^*$ is the initiator satisfies the following properties. 

\vspace{2mm}
\noindent \textbf{Validity.} If a rational agent $p^*$ broadcasts a message $m$, then $p^*$ eventually delivers $m$. 

\noindent \textbf{Agreement.} If a rational agent delivers a message $m$, then all rational agents eventually deliver $m$. 

\noindent \textbf{Integrity.} For any agent $p^*$ and sequence number $s$, every rational agent delivers at most one message $m$ initiated by $p^*$ with sequence number $s$. Moreover, this delivery only occurs if $m$ was previously broadcast by $p^*$.
\vspace{2mm}

These properties mimic those of BRB \cite{Hadzilacos1994} except that they provide guarantees with respect to rational agents rather than correct agents.
To achieve these properties, we adopt a mechanism used by \cite{aiyer2005bar} with some minor modifications. 
In particular, agents impose the following localized penalty in reliable broadcast communications against other agents who deviate from the correct protocol execution.

\begin{definition}[Retaliation]
 Consider an instance of RRB called $B$. Suppose $B$ is tasked with broadcasting a message, and has a sequence number $s$ internal to its instance of the protocol. If an agent $p_i$ does not receive a message from another agent $p_j$ within the execution of $B$ for sequence number $s$, then $p_i$ \emph{retaliates} by stopping communications with $p_j$ in protocol instance $B$ for any future sequence number $s' > s$ until $p_i$ receives the missing messages from $p_j$ for instance $B$, sequence number $s$. \label{def:retaliation}
\end{definition}

This mechanism is intended to punish lazy agents who attempt to reduce their communication costs in the reliable broadcast protocol. 
Observe that without this mechanism, lazy agents may listen to incoming messages but never send any outgoing messages.

\subsection{RRB Protocol Description}

The full Rational Reliable Broadcast protocol is essentially the protocol from \cite{bracha1987asynchronous}, with a retaliation penalty added using a message queuing system. 
This queuing system updates vectors tracking the sequence numbers of the last messages received by each other agent and delays sending messages out about a sequence number $s$ to a certain agent $p$ until all messages regarding sequence numbers up to and including $s - 1$ have been processed. 
For the sake of self-containment, we reproduce the full algorithm in Algorithm \ref{algo:rrb}. For brevity, presentation is deferred to Appendix \ref{app:rrbalg}

In words, Algorithm \ref{algo:rrb} describes the following: An agent $p^*$ initiates a broadcast in the RRB instance by marking their message with the tag \emph{initial} and sending it to each other agent. 
An agent receiving a message \emph{msg} forwards the contents with the tag \emph{echo} after one of two conditions occur. 
Either the agent received \emph{msg} from $p^*$ itself in an initial message, or it has received enough \emph{echo} or \emph{ready} messages to ensure that $p^*$ was the initial sender of \emph{msg}. 
An agent sends a \emph{ready} message when it knows that \emph{msg} is the only message broadcast by $p^*$ for that sequence number. 
This can be guaranteed because the first agent to send a \emph{ready} message must have seen $(n+t)/2$ \emph{echo} messages, and two separate agents cannot receive this many \emph{echo} messages for distinct message contents because that would imply that a correct agent sent multiple \emph{echo} messages for different message contents. 
This portion of the algorithm is essentially unchanged from \cite{bracha1987asynchronous}, so the reader can find more information there.

To enable retaliation against lazy processes, a queuing system is used to process messages in order of sequence number. 
This queuing system allows agent to detect if another agent has neglected their duties of participating in an instance and exclude this agent from their communications within this instance until the missing messages are received. 
Separate queues are used for incoming \emph{echo}, incoming \emph{ready}, and outgoing messages. 
When an \emph{echo} or \emph{ready} message is received, it is not processed until all messages with lesser sequence numbers are cleared from the queue. 
Likewise when the outgoing queue is flushed, outgoing messages to agents that have not sent messages for prior sequences is delayed until this is remedied. 
Queues are flushed whenever an agent's knowledge changes on receiving a message. 
Within the flushing routine are blocks to ensure multiple messages from an agent for the same sequence number are ignored. 
Within the protocol are also checks to ensure that a correct agent does not send duplicate messages, so these blocks ensure that correct agents do not exhibit deviant behavior such as sending multiple \emph{initial} messages with different contents for the same sequence number. 
\footnote{We remark that this retaliation mechanism does impact performance, as we are forcing agents to react to messages from other agents in order of sequence numbers. 
An interesting direction for future work would be to examine this performance impact or design alternate incentive mechanisms for reliable broadcast in the presence of rational agents.}

\subsection{RRB Correctness}

To show the correctness of RRB, we first show that the retaliation penalty in Definition \ref{def:retaliation} is sufficient to guarantee that all rational agents behave correctly so long as they believe that all other rational agents will also behave correctly. 
We first make the following assumption:

\begin{assumption}\label{ass:prefer}
Agents have a preference towards a $(2t+1)$-sized quorum signalling ``ready'' over not in Rational Reliable Broadcast, due to the context the protocol is being used in.
\end{assumption}

When we refer to the context of the protocol's deployment, we mean whether the use case implies rewards for participation or is directly linked to utility gain such as in a money transfer system.
It is easy to argue that this assumption is achievable in Hearsay.
In Hearsay, all agents are given a reward in the form of a transaction fee when a transaction is executed.
As we will see in Section \ref{sec:hearsay}, a transaction cannot be executed before it is delivered, which does not occur until a quorum is reached. Therefore, in our context, agents will prefer these quorums to form as otherwise they may be blocked from reward, and even correctly interacting with the broadcast (i.e. send and receive funds) with the initiator of the broadcast channel. It is for this reason we have the following conditional result:

\begin{corollary}\label{cor:prefer}
Predicated on Assumption \ref{ass:prefer}, rational agents prefer to be a member of $(2t+1)$-sized quorums signalling ``ready'' in Rational Reliable Broadcast
\end{corollary}

\begin{proof}
This is due to a ``fear of missing out'' on successful quorums. 
This is because in the worst case (i.e., when there are $t = \lceil \frac{N}{3} \rceil - 1$ Byzantine agents who choose not to participate in the RRB instance), a rational agent's choice to abstain from participating in a broadcast will make a $(2t+1)$-sized quorum impossible to achieve. Suppose the worst case occurs with nonzero probability.
During any execution of Rational Reliable Broadcast, a rational agent is tasked with reasoning about whether it is better to abstain to lessen the costs of communication or join the quorum to ensure they receive positive utility. 
In the worst case, a rational agent is pivotal in ensuring a quorum occurs, so a reward in the context of Assumption \ref{ass:prefer} can be made large enough such that the expected utility, when considering the probability of the worst case outcome, is larger than any costs against participating in the quorum. This gives the result.

Furthermore, suppose we impose an assumption similar one used in \cite{aiyer2005bar}, which assumes that rational agents assign high weight to the probability of worst case behavior (i.e. rational agents are \emph{risk averse}). This would be because they value the benefits of the service and view the risk of service failure as unacceptable, which is not far-fetched in infrastructural technologies such as monetary systems\footnote{Moreover, notice that behavior such as the formation of mining pools (see \cite{miningpools}) in cryptocurrencies suggests that participants in distributed monetary systems are naturally risk averse, as this widespread strategy is done specifically to reduce risk.}. Then this precludes the quantity of a suitable reward from being prohibitively large.
\end{proof}

We now show that rational agents behave according to the protocol in RRB.
\begin{lemma}\label{lem:rrb}
No rational agent has a unilateral incentive to deviate from the RRB protocol.
\end{lemma}
\begin{proof}
Fix a RRB instance $B$, initiating agent $p^*$, and sequence number $s$. First notice that no agent has any incentive to send extra \emph{echo} and \emph{ready} messages: Indeed lines \ref{algoline:ignore1} and \ref{algoline:ignore2} ensure correct agents ignore such behavior in a single sequence, so conditioning on the other rational agents following the protocol, no agent can gain anything by sending extra \emph{echo} and \emph{ready} messages. Similar logic applies for sending multiple \emph{initial} messages for a single sequence, as analysis in \cite{bracha1987asynchronous}, and given intuitively in the description of the algorithm, show that given the other rational agents are behaving correctly, only one of the initial messages will be accepted by the rational agents. Therefore we are interested in showing that rational agents cannot benefit by abstaining from forwarding messages in the protocol (i.e. we prove there are no \emph{free riders}).

Let $p_{\emph{bad}}$ be a deviating agent which does not send messages to agent $p_\emph{victim}$ for the reliable broadcast initiated by $p^*$ for sequence $s$. By the retaliation policy, $p_\emph{victim}$ will refuse to send messages to $p_\emph{bad}$ about broadcasts from initiator $p^*$ with sequence number greater than $s$. In the worst case, $p_\emph{bad}$ will be unable to participate in quorums of $2t+1$ ``ready'' signals on any future transaction from $p^*$. By Corollary \ref{cor:prefer} any rational agent would prefer to not be in this situation. 
Therefore a rational agent has no unilateral incentive to deviate from the RRB protocol.
\end{proof}
\begin{remark}
When we argued Corollary \ref{cor:prefer} , we pointed out that in the worst case a rational agent $p_\emph{bad}$ choosing to abstain from participation in the RRB protocol could lead to a ``stalled'' broadcast instance for the initiating agent $p^*$. In this case the retaliation other agents enact on $p_\emph{bad}$ is in some ways giving $p_\emph{bad}$ control over the broadcast capabilities of $p^*$, and $p_\emph{bad}$ could ransom these capabilities to force $p^*$ to pay them to un-stall their channel. Under the assumption that no externalities, such as extra communication channels, are present among rational agents, this attack is impossible. It would be interesting future work to examine what dynamics may arise when such externalities are present.
\end{remark}
Now we can prove the protocol is correct.
\begin{theorem}\label{thm:rrb}
The Rational Reliable Broadcast protocol is correct, i.e. it satisfies the \emph{validity}, \emph{agreement}, and \emph{integrity} properties.
\end{theorem}
\begin{proof}
The only changes to the BRB protocol given in \cite{bracha1987asynchronous} are the queuing and retaliation systems. 
As we assume reliable (but still asynchronous) communication, the queuing system should then have no impact on the activities of the rational agents. 
Likewise the retaliation system will have no impact on the activities between rational agents, as Lemma \ref{lem:rrb} shows that rational agents behave correctly. 
Therefore each of the properties is inherited from \cite{bracha1987asynchronous}. In particular, \emph{validity} and \emph{agreement} are shown by property 1, which is proven in Theorem 1 and Lemma 4 in \cite{bracha1987asynchronous}. \emph{Integrity} simply follows by Algorithm \ref{algo:rrb}, as duplicate messages are ignored and the delivery of a message can only happen once for each sequence number. 
Additionally, \ref{line:echoreadycond} of Algorithm \ref{algo:rrb} states that a correct agent must see more \emph{echo} or \emph{ready} messages than there are Byzantine agents, meaning only messages correctly initiated by the sender are able be delivered.
\emph{Integrity} follows. 
\end{proof}

\section{The Hearsay Protocol}
\label{sec:hearsay}

In this section, we describe the Hearsay protocol and provide pseudo-code in Algorithm \ref{algo:protocol}. 
Hearsay operates with $N$ fully-connected and uniquely identified agents, and each agent maintains a public-private key pair for signing and authenticating messages. 
The only messages that are broadcast in Hearsay are transactions messages, in which the agent initiating the broadcast specifies an amount of money to be transferred to a single recipient agent.
Additionally, the initiating agent must pay a fixed transaction fee $\epsilon$ to each other agent in the system upon execution of a transaction. 

The state of the system is replicated across each agent in the form of local variables $\mathcal{B}$, a vector of balances for each agent, and $\mathcal{S}$, a list of agent sequence numbers which allow the outgoing transactions of each agent to be totally ordered. 
Additionally, each agent maintains a set of incoming payment buffers and a set of fee buffers, which contain incoming payments and fees, respectively, for each agent which have not yet been credited to their balance. 
When an agent initiates a new transaction, they can specify the incoming payments and fees which must be added to their balance prior to executing the transaction. 

\subsection{Rational Reliable Broadcast in Hearsay}
Hearsay makes use of $N$ instances of the RRB protocol to enable broadcasts between agents. 
All agents participate in all of the RRB instances, but each agent is only allowed to initiate a broadcast in a single instance.

In order to maintain modularity between Hearsay and the underlying broadcast primitive, we introduce the following vocabulary. 
An agent \textit{broadcasts} a transaction if they construct the transaction, sign it, and send it to all other agents.
An agent $p$ \textit{receives} a transaction if it is sent to them during execution of an RRB instance.
After achieving a sufficient quorum on a received transaction to satisfy the broadcast primitive, $p$ then \textit{delivers} the transaction by adding it to a buffer containing transactions to be executed after a set of conditions are satisfied.
Once the required conditions are met, $p$ \textit{executes} a delivered transaction by running the \textbf{Execute} function.
Finally, if $p$ determines that the transaction is not a \emph{bad} transaction, then $p$ \emph{commits} the transaction by updating the balances of the initiator and recipient accordingly.

This modularity allows messages to be delivered in each RRB instance without any constraints from the Hearsay protocol while ensuring that transactions are executed in accord with certain dependency conditions (as specified in Algorithm \ref{algo:protocol}).

\subsection{Protocol Description}

Money transfer begins when an agent constructs a transaction with the \textbf{Pay} function. 
A transaction consists of the following elements: $1$) the index of a source agent $i$ from which money will be deducted, $2$) the index of a destination agent $j$ which will receive the money, $3$) an expenditure amount $x$ that will be transferred from agent $i$ to agent $j$, $4$) the sequence number of agent $i$'s outgoing transaction, and $5,6$) a set of recent payments and fees received by agent $i$ which are used to verify that $i$'s balance is sufficient to make the transfer. 

After composing a transaction, $i$ broadcasts it using the RRB primitive.
Each agent delivers any transaction for which they receive a RRB quorum. 
In the case that a transaction never achieves a RRB quorum, it is clear by the Validity property of RRB that the initiator of the transaction must be Byzantine, and the transaction therefore does not need to be delivered.

Upon delivering a transaction, an agent places the transaction into a buffer until the following conditions are met.
The first condition is that, in order for an agent to execute a transaction with initiator $p$ and sequence number $s$, their local sequence number for initiator $p$ must be $s-1$. 
This ensures that any correct agent executes all transactions from $p$ in order of sequence number.
The second condition which must be met before a transaction can be executed is that all of its dependent transactions must be executed. 
Because communication is asynchronous and agents may execute transactions in different order, the second condition is imposed to guarantee that any money received by the initiator of a transaction which is required for the initiator to pass a balance check is indeed added to their balance before their transaction is executed. 
The third and final condition which must be met before a transaction may be executed is that the initiator of the transaction must either have a balance of $N \epsilon$ prior to the execution of the transaction or the transaction must convert at least $N$ fee credits into balance. 
This condition guarantees that no agent can achieve a negative balance (i.e., go into debt). 
Without checking the third condition prior to execution, any agent with a balance of less than $N \epsilon$ which initiates a transaction will end up with a negative balance, as fees are deducted from the initiator's balance regardless of whether or not the transaction is committed. 

Agents execute a transaction by updating their local variables according to the \textbf{Execute} function. 
While executing a transaction, each agent must verify that the agent which initiated the transaction has sufficient balance to cover both the specified payment and the required transaction fees.
If an agent executing a transaction determines that the initiating agent does not have sufficient balance (i.e., the transaction overdrafts), the transaction is marked as \emph{bad}. 
Additionally, any transaction which references a bad transaction in its dependency set (meaning a bad transaction is included in the set \emph{tx.deps}) is also marked as bad. 
The punishment for a bad transaction is that the transaction fee is transferred from the balance of the initiating agent to the balance of each other agent, but the payment specified by the transaction is not transferred from the initiator to the recipient. 
In other words, the fee is charged as a punishment but the transaction itself does not occur. 
If an agent decides that a transaction is not bad, then they reduce the balance of the initiator and add the transaction to the recipient's incoming payment buffer via the \textbf{Commit} function. 


\begin{algorithm}[t]
\scriptsize
\DontPrintSemicolon
\SetKwInOut{Input}{input~}
\SetKwInOut{Local}{Local variables~}
\SetKwInOut{Output}{output~}
\DontPrintSemicolon
\Local{$N$ :: the number of agents in the system \\
$\epsilon$ :: the transaction fee \\
$i \in \{1, \dots, N\}$ :: my agent number \\
$F = \{\}^N$ :: set of $N$ fee buffers \\
$Q = \{\}^N$ :: set of $N$ incoming payment buffers \\
$B \in \mathbb{R}^N$ :: agent balances, initially all greater than $N \epsilon$  and agreed upon by all agents \\
$S = 0^N$ :: agent sequence numbers, initially 0 \\
$\emph{ctr} = 1$ :: internal transaction counter, initially 1 \\
}
\vspace{\baselineskip}
\textbf{Pay}(Recipient $j$, Amount $x$, Fee Conversion Flag $f$): \;
\algotab \uIf{$B[i] < x + N \epsilon$\label{hearsay:check_my_balance}\label{algoline:fee_sufficient}}{
\algotab \textbf{exit}
}
\algotab \uIf{f = true\label{algoline:fee_flag_check}}{
\algotab \emph{fees} $\leftarrow F[i]$ \;
\algotab $F[i] \leftarrow \{\}$\;
}
\algotab \uElse{
\algotab \emph{fees} $\leftarrow \{\}$ \;
}
\algotab (\emph{tx}.\emph{initiator}, \emph{tx}.\emph{recipient}, \emph{tx}.\emph{value}, \emph{tx}.\emph{seq}, \emph{tx}.\emph{deps}, \emph{tx}.\emph{fees}) $\leftarrow (i, j, x, \emph{ctr}, Q[i], \emph{fees})$ \label{hearsay:assign_tx}\label{algoline:tx_construction}\;
\algotab $Q[i] \leftarrow \{\}$ \label{hearsay:assign_I}\label{algoline:clear_tx_buffer}\;
\algotab $\emph{ctr} \leftarrow \emph{ctr}  + 1$ \label{hearsay:assign_c}\;
\algotab \textbf{RRB}(\emph{tx}) \;
\;

On \textbf{RRB-Deliver}(\emph{tx}): \;
\algotab add \emph{tx} to execution buffer \label{hearsay:buffer}\;
\;

\textbf{Condition 1:} $S$[\emph{tx}.\emph{initiator}] = \emph{tx}.\emph{seq} - 1 \label{hearsay:cond1}\;
\textbf{Condition 2:} $\forall t \in$ \emph{tx.deps} $\cup$ \emph{tx.fees}: $t$ executed \label{hearsay:cond2}\;
\textbf{Condition 3:} $(B[\emph{tx.initiator}] > N \epsilon) \lor (|\emph{tx.fees}| > N)$ \label{hearsay:cond3}\;

for transaction satisfying Conditions 1 and 2 \textbf{Execute}(\emph{tx}): \;
\algotab \emph{tx}.\emph{bad} $\leftarrow$ \emph{false} \;
\;

\algotab \For{$t \in \text{tx.deps}$}{\label{hearsay:start}
\algotab \uIf{$t \in Q[\text{tx.initiator}]$}{
\algotab $B[\emph{tx.initiator}] \leftarrow B[\emph{tx.initiator}] + \emph{t.value}$ \label{hearsay:add_payments}\;
\algotab remove t from Q[\emph{tx.initiator}] \;
}
}
\;

\algotab \For{$t \in \text{tx.fees}$}{
\algotab \uIf{$t \in F[\text{tx.initiator}]$}{
\algotab $B[\emph{tx.initiator}] \leftarrow B[\emph{tx.initiator}] + \epsilon$ \label{hearsay:add_fees}\;
\algotab remove t from F[\emph{tx.initiator}] \;
}
}
\;

\algotab \uIf{($B$[tx.initiator] $< x + N\epsilon$) $\lor$ ($\exists t \in$ tx.deps : t.bad = true \label{hearsay:bad_cond}\label{algoline:bad_tx_check})}{
\algotab \emph{tx}.\emph{bad} $\leftarrow$ \emph{true} \label{hearsay:bad_assign}
}

\; 

\algotab \uIf{tx.bad = false}{
\algotab \textbf{Commit}(\emph{tx})
}
\; 

\algotab $B$[\emph{tx}.\emph{initiator}] $\leftarrow$ $B$[\emph{tx}.\emph{initiator}] $- N \epsilon$ \algotab (fees occur whether or not the transaction is bad) \label{hearsay:subtract_fees}\;
\algotab \For{$j \in \{1, \dots, N\}$}{
\algotab $F[j] \leftarrow F[j] \cup tx$ \label{hearsay:add_feedependency}\label{algoline:fee_awarding}\;
}
\;

\algotab $S$[\emph{tx}.\emph{initiator}] $\leftarrow$ $S$[\emph{tx}.\emph{initiator}] $+ 1$ \label{hearsay:update_sequence}\;
\;

\textbf{Commit}(tx): \\
\algotab $B$[\emph{tx}.\emph{initiator}] $\leftarrow$ $B$[\emph{tx}.\emph{initiator}] - \emph{tx.value} \label{hearsay:subtract_payment}\;
\algotab $Q[\emph{tx.recipient}] \leftarrow Q[\emph{tx.recipient}] \cup \emph{tx}$ \label{hearsay:add_txdependency}\;
\caption{Hearsay}
\label{algo:protocol}
\end{algorithm}

\subsection{Transaction Fees and Balance Updates}
For each transaction executed in Hearsay, the initiator of the transaction must pay a fixed fee $\epsilon$ to each agent in the system in return for executing the transaction. 
As we describe in Section \ref{sec:incentives}, this fee is the main incentive by which rational agents are encouraged to follow the Hearsay protocol.
Due to the asynchronicity of communications in the system, the order in which payments and fees are credited to agent balances must be strictly defined in order to ensure that any two agents executing the same transaction are able to agree on the balance of that transaction's initiator.

When a correct agent $p$ executes a transaction, $p$ deducts $N \epsilon$ directly from the balance of the initiator of the transaction. 
However, rather than immediately incrementing the balance of every agent in the system by $\epsilon$, $p$ instead adds a fee \textit{credit} to a fee buffer for each agent (including $p$ itself). 
In order for $p$ to convert their fee credits into a spendable balance, $p$ must signal a fee conversion to occur with the next transaction they initiate. 
Specifically, $p$ must reference a set of recent transactions for which $p$ has received a fee credit within their next transaction $tx$.
This signals to each executing agent to treat all fees referenced by $tx$ as dependencies which must be executed prior to $tx$.
Then, upon execution of $tx$, each of the fee credits referenced is converted into balance for $p$. 

Incoming payments to $p$ are treated similarly to fees. 
If $p$ is the recipient of a transaction, each agent executing the transaction adds the transaction value to the incoming payment buffer for $p$ instead of immediately adding it to $p$'s balance. 
Each time $p$ initiates a transaction $tx$, they reference the set of incoming payments they have received since their last outgoing transaction. 
These are also treated as dependencies  which must be executed prior to $tx$, and the balance of each incoming payment referenced is added to $p$'s balance upon execution of $tx$.

If fees and incoming payments were added immediately to the recipient's balance, then two agents which have executed different sets of transactions which are not dependencies of $tx$ might have added different fees and payments to $p$'s balance prior to executing $tx$.
As a result, these two agents might disagree on whether or not to commit $tx$, breaking the agreement property.

We also remark that in general, a transaction which signals a fee conversion is executed more slowly than one which does not, as a fee conversion adds many additional dependencies which must be executed prior to the transaction.
For this reason, fee conversions should be treated as rare events.
Although an agent may choose to signal a fee conversion at any time, it makes sense for a rational agent to wait until they need the balance for a future transaction or they know that they will not need to initiate any transactions in the near future.

\subsection{Protocol Properties}
\label{sec:hearsay_properties}
Let Hearsay be the protocol characterized by the pseudo-code in Algorithm \ref{algo:protocol}. 
In this section we show that Hearsay satisfies the following properties.
Let the set of rational agents be $\mathcal{R}$.

\vspace{2mm}
\noindent \textbf{Validity}. If a rational agent $p$ broadcasts a transaction \emph{tx}, then $p$ eventually executes \emph{tx}.\\

\textit{Formally}. 
For transaction \emph{tx}:

\begin{equation}
    \forall p \in \mathcal{R}, p \textbf{ broadcasts } tx \rightarrow p \textbf{ executes } tx
\end{equation}

\noindent \textbf{Agreement}. If a rational agent $p$ executes a transaction \emph{tx}, then all rational agents eventually execute \emph{tx}. 
Moreover, if $p$ commits $tx$, then all rational agents commit $tx$.

\textit{Formally}. 
For a transaction \emph{tx}:

\begin{equation}
\begin{split}
    \big((\exists p \in \mathcal{R} : p \textbf{ executes } tx) &\rightarrow (\forall q \in \mathcal{R}, q \textbf{ executes } tx)\big)\\
    \land \big((\exists p \in \mathcal{R} : \textbf{ commits } tx) &\rightarrow (\forall q \in \mathcal{R}, q \textbf{ commits } tx)\big)
\end{split}
\end{equation}

\noindent \textbf{Integrity}. For any initiator $p$ and sequence number $s$, any rational agent executes at most one transaction with initiator $p$ and sequence number $s$.\\

\textit{Formally}. 
Let $tx_1, tx_2$ be any two transactions. 

\begin{equation}
\begin{split}
    &(\exists p \in \mathcal{R} : p \textbf{ executes } tx_1 \land p \textbf{ executes } tx_2) \rightarrow (tx_1.\emph{initiator} \neq tx_2.\emph{initiator} \lor tx_1.\emph{seq} \neq tx_2.\emph{seq})  \\
\end{split}
\end{equation}


In order to characterize agent behavior in a setting where a rational agent may choose to deviate from Hearsay at any time, we introduce the notion of past-correctness. 
However, for the sake of clarity, we simply call an \textit{correct} rather than past-correct.

\begin{definition}[Past-correct]
An agent $p$ is said to be past-correct at an event $e$ iff at all local events before and including $e$, $p$ has acted according to Hearsay exactly. 
\end{definition}

We now introduce the first major theorem of the paper, which states that Hearsay satisfies the properties of Validity, Agreement, and Integrity if all rational agents behave correctly. 
In Section \ref{sec:incentives}, we introduce the other major theorem of the paper, which states that all rational agents behave correctly. 

\begin{theorem} Hearsay satisfies the properties of Validity, Agreement, and Integrity if no rational agent deviates from the protocol specification.
\label{thm:properties}
\end{theorem}

\vspace{-1mm}
\noindent The intuition behind the Theorem is that, if we assume that rational agents behave correctly, if a rational agent broadcasts a transaction then eventually all rational agents execute it and update their balances in precisely the same way. 
We prove Theorem \ref{thm:properties} by proving the next three lemmas. 
We defer the formal proofs of each lemma to Appendix \ref{app:defproofs} and provide more intuitive proof sketches in this section.

\begin{lemma} Hearsay satisfies Validity if no rational agent deviates from the protocol specification. \label{lem:validity}
\label{lem:1}
\end{lemma}

\begin{proof}[Proof Sketch]
We first state that if a rational (and therefore correct) agent $p$ has executed all of their own transactions prior to $tx$, and $p$ broadcasts a transaction \emph{tx}, then $p$ will eventually deliver \emph{tx}.
We prove this statement with the following logic.
By the Validity of RRB, if $p$ broadcasts $tx$, then $p$ delivers $tx$.
By hypothesis, $p$ has executed all of their transactions prior to $tx$ (satisfying Condition 1).
By the Hearsay specification, if $p$ broadcasts $tx$, then they must necessarily have executed all transactions that they placed in \emph{tx.deps} (satisfying Condition 2).
And again by the Hearsay specification, if $p$ broadcasts $tx$, their balance must be greater than $N \epsilon$ (satisfying Condition 3).

The proof of the lemma then follows by simple induction.
Before $p$ has broadcast their first transaction $tx$, it trivially holds that $p$ has delivered all of their own transactions prior to $tx$. 
We then use induction to show that $p$ will also execute each transaction that they broadcast after $tx$. 
\end{proof}

\begin{lemma} Hearsay satisfies Agreement if no rational agent deviates from the protocol specification.\label{lem:agreement}
\end{lemma}
\begin{proof}[Proof Sketch]
We first define the set $\Phi(tx)$, which is constructed recursively on a transaction $tx$ with initiator $p$ to include $tx$, all transactions initiated by $p$ with sequence number less than or equal to $tx$, and any dependency of any transaction in $\Phi(tx)$.
Intuitively, if a correct agent $p$ executes $tx$ then $\Phi(tx)$ is the set of all transactions that $p$ \emph{must also} have executed.
We also define the height of $tx$ to be the length of the longest chain of transactions in the set $\Phi(tx)$. 
Intuitively, the height of $tx$ is the longest chain of dependencies a correct agent must execute prior to $tx$.

The proof uses an induction argument over transaction heights. 
It is clear that any transaction of height 1 has no dependencies and the initiator must have sufficient balance to cover the fees, so the transaction is executed immediately after it is delivered by each rational agent and all rational agents agree on whether or not to commit it. 
This forms the base case for induction. 
As an induction hypothesis, we assume that all transactions up to height $k$ have been executed by all rational agents and that all rational agents agree upon whether not to commit them.
The induction hypothesis implies that for some transaction $tx$ of height $k+1$, all transactions in $\Phi(tx)$ other than $tx$ must have been executed by all rational agents and that all rational agents have agreed upon whether not to commit them.
We combine this with the observation that the balance of the initiator of $tx$ can only be changed prior to the execution of $tx$ by a transaction in $\Phi(tx)$.
This implies that all rational agents must have performed the same updates to the balance of $\emph{tx.initiator}$ prior to the execution of $tx$, and therefore all rational agents agree on the balance. 
Finally, during the execution of $tx$, all rational agents are guaranteed to perform the same balance updates prior to the balance check because they are guaranteed to have executed all dependencies of $tx$ by the induction hypothesis. 
This necessarily means that all rational agents execute $tx$ and agree upon whether or not to commit it, completing the proof.

The intuition behind the proof is that if some rational agent executes $tx$, they must also have executed all transactions in $\Phi(tx)$. 
Then for any other rational agent which has not yet executed $tx$, there must be at least one transaction in $\Phi(tx)$ that they can execute.
Moreover, because the balance of any agent $p$ only changes during a transaction that they initiate, the order in which transactions not initiated by $p$ are executed does not affect the balance checks for $p$'s transactions. 
As a result, all rational agents execute the same transactions and agree upon whether or not to commit each executed transaction.
\end{proof}

\begin{lemma} Hearsay satisfies Integrity if no rational agent deviates from the protocol specification. \label{lem:integrity}
\label{lem:4}
\end{lemma}
\begin{proof}[Proof Sketch]
The proof of this lemma follows almost directly from the Integrity of RRB. 
If a rational agent $p$ executes two transactions with the same initiator and sequence number, it must also be true that $p$ delivered both transactions. 
However, by the Integrity of RRB, $p$ will not deliver two different transactions with the same initiator and sequence number.
\end{proof}

It is clear that the proofs of Lemmas \ref{lem:validity}, \ref{lem:agreement}, and \ref{lem:integrity} constitute a proof of Theorem \ref{thm:properties}. 

\subsection{Implications of Theorem \ref{thm:properties}}
Theorem \ref{thm:properties} states that if no rational agent deviates from the protocol, Hearsay satisfies Validity, Agreement, and Integrity. 
These properties guarantee that Hearsay maintains the necessary functionality for a distributed money transfer system.
Specifically, Theorem \ref{thm:properties} has the following corollaries. 

\begin{corollary}
Any transaction broadcast by a rational agent is eventually executed by all rational agents. 
\end{corollary}

This statement follows directly from the combination of Validity and Agreement and guarantees that rational agents are able to spend their money. 
Without this property, the arguments used in Section \ref{sec:incentives} do not hold, as rational agents do not benefit from receiving money that they are unable to spend. 
However, note that this property need not hold for Byzantine agents which broadcast conflicting transactions with the same sequence number. 
In this case, it may be that neither transaction is ever executed.

On the other hand, agents should also not be able to spend money that they do not own and double-spends should be impossible.
If an agent could reliably spend more money than they had, this too would break down the arguments in Section \ref{sec:incentives}, as rational agents have a lower incentive to execute the protocol in order to acquire money. 

\begin{corollary}
No agent may spend money they do not have or enter into debt. 
Moreover, no agent may spend the same money twice.
\end{corollary}

This corollary follows from Agreement and Integrity.
By Agreement, if there is any rational agent which determines that the initiator of a transaction does not have sufficient balance to complete the payment, then no rational agent commits the transaction. 
If the initiator does not have sufficient balance to cover the transaction fees, then no rational agent even executes the transaction, again ensuring that no agent may enter a negative balance.
By Integrity, if no rational agent may execute two transactions from the same initiator with the same sequence number. 
This implies that each time transaction is executed, the money is deducted from the initiator's balance prior to the execution of their next transaction.

\section{Incentive Analysis}
\label{sec:incentives}

Because we are interested in systems with a large number of rational agents, we cannot simply prove the Hearsay protocol is correct when at least $\frac{2}{3}$ of the agents behave correctly.
We must instead show that the rational agents achieve the maximum possible reward by following the protocol correctly.
If this is the case, then rational agents will see no reason to deviate.
In particular, in this section we show that following the Hearsay protocol is a Nash equilibrium amongst the rational agents.
This is stated in the following theorem.

\begin{theorem}
No rational agent has any incentive to deviate from Hearsay, assuming the other rational agents are following the protocol.
\label{thm:correqhearsay}
\end{theorem}


We prove this statement through the following two lemmas. 
The first lemma states that no rational agent has incentive to deviate from the prescribed \textbf{Pay} function, and the second states that no rational agent has incentive to \emph{execute} transactions incorrectly. 
The theorem is corollary of these two lemmas. 
We periodically reference the line numbers of Algorithm \ref{algo:protocol}.

\begin{lemma}
No rational agent has any incentive to deviate from the \textbf{Pay} function in Hearsay, assuming the other rational agents are following the protocol.
\end{lemma}
\begin{proof}
Suppose $p$ is a rational agent executing the \textbf{Pay} function in Algorithm \ref{algo:protocol}, as they wish to send a payment to another agent or convert their fees into their balance. 
Note that $p$ will initiate this transaction as it is safe to assume $p$ prefers creating a transaction over not creating it so long as the system chooses a suitable transaction fee.
First see that $p$ will not choose to initiate a transaction if their balance is not sufficient (line \ref{algoline:fee_sufficient}), and Condition 3 for their transaction is therefore satisfied (line \ref{hearsay:cond3}).
If $p$ were to do so, the other rational agents would mark the transaction as ``bad'', thereby not carrying out the transaction yet still collecting the fee from $p$ anyways (lines \ref{algoline:bad_tx_check}-\ref{algoline:fee_awarding}).
For the same reason, $p$ will ensure that they construct their transaction correctly (lines \ref{algoline:fee_flag_check}-\ref{algoline:tx_construction}). 
In particular, it is clear to see that $p$ will not use a ``bad'' transaction in their dependency list due to fee collection and will specify the correct recipient so as to not risk sending funds to the wrong destination. 
Additionally, $p$ will maintain the correct sequence number and dependencies to ensure that other rational agents are able execute the transaction according to Conditions 1 and 2 (lines \ref{hearsay:cond1} and \ref{hearsay:cond2}).
Next, it is clear that $p$ will clear their transaction buffer (line \ref{algoline:clear_tx_buffer}), because they'd prefer not to participate in unnecessary communication.
And finally $p$ must participate in a RRB broadcast to communicate their transaction to the other agents, as other rational agents who are assumed to behave correctly will not execute a transaction until it is \emph{delivered} in RRB.
Thus $p$ will not deviate in the \textbf{Pay} function.
\end{proof}

\begin{lemma}
No rational agent has any incentive to \emph{execute} a transaction incorrectly (or not at all) in Hearsay, assuming the other rational agents are following the protocol.
\end{lemma}
\begin{proof}
Here we examine the effects on $p$ that result from $p$ not executing a transaction as specified by the protocol. 
As the purpose of the \textbf{Execute} function is to reward fees and maintain local knowledge of balances, the following argument details the potential outcomes of deviating in the \textbf{Execute} function and obtaining an incorrect local view of the balances within the system. Note, Theorem \ref{thm:properties} and Lemma \ref{lem:agreement} tell us that the \emph{Agreement} property is satisfied if no rational agent deviates from Hearsay. 
This means that if a correct agent executes a transaction, then all correct agents will execute the transaction.
Then $p$ may have a local view that is inconsistent with the other rational agents in the system if $p$ does not execute the transaction correctly.
It is clear that incorrectly executing a transaction in a way that gives extra funds to themselves (i.e. making money appear out of thin air) would not benefit $p$. 
Every rational agent maintains their own local view of $p$'s balance, so if those extra funds become pivotal to a transaction $p$ initiates, then each rational agent will penalize $p$ by deducting fees for the transaction without committing it.
Furthermore an agent will always award funds to itself correctly so that it can receive payments or rewards.
This tells us two important details:
\begin{itemize}
    \item An agent cannot gain utility through awarding itself extra funds by any means. Therefore utility can only be gained by abstaining from execution to eradicate computational costs.
    \item Any incorrect execution (or lack of execution) that is rational would result in $p$'s local view of another agent $q$ being either too high or too low
\end{itemize}
The second point provides two cases for analysis.

Suppose that an agent $q$ pays $p$ in a transaction when $p$ thinks $q$'s balance is too low (because $p$ decided it was worthwhile to avoid executing the dependencies), but the rest of the system is assured it is not.
Then in this case, $p$ will mark $q$'s transaction as ``bad" and elect not to commit it. 
As a result, $p$ will not have knowledge of funds they have received and lose out on the opportunity to safely spend those funds.

Now assume the case that $p$'s local view of the balance of some agent $q$ is too high.
Suppose $q$ is Byzantine and sends a payment to $p$ that is larger than the balance $q$ actually has but smaller than $p$ believes $q$ has. 
Then $p$ will accept the payment to obtain more funds while the other rational agents mark this transaction as bad and do not commit it. 
So, now $p$ has a local view of their own balance that is higher than that of the rest of the system. 
This means that the other rational agents will penalize $p$ by deducting fees without committing $p$'s transaction if these extra funds become pivotal.

Similar to Section \ref{sec:rrb}, we can apply the assumption that agents value the benefits of the service and perceive risk of service failure negatively.
Then the nature of the above scenarios would prevent rational agents from deviating from the protocol for the benefit of saving on computational resources.
Therefore in both cases, $p$ prefers to execute the transaction correctly.
\end{proof}

\section{Discussion and Future Work}\label{sec:future}

In this work we presented Hearsay, a novel distributed monetary system that tolerates both Byzantine and rational behavior and operates without the need for consensus. 
To accomplish this we proposed a new reliable broadcast primitive called Rational Reliable Broadcast, an incentive-based algorithm solving a classical problem at the intersection of game theory and distributed computation. 
We hope our work inspires further inquiry, and we discuss some possible future work below.

\paragraph{Externalities and Stronger Game Theoretic Guarantees} In Section \ref{sec:rrb} we pointed out that Rational Reliable Broadcast may fail if we consider attacks where rational agents can ``ransom'' a broadcast channel through external communication channels. Interesting dynamics arise if we allow such externalities. The presence of externalities such as collusion, extra communication, or inter-agent relationships like spite opens the floodgates to complex strategies, so we assumed in this preliminary work that they do not occur.
However, detailed investigation of the effects of these phenomena in our algorithms and distributed computation at large could lend interesting results. 
In particular, providing stronger game theoretic guarantees for our algorithms is a direct extension of our work. 
For example, our model is very coarse-grained in the sense that any agent exhibiting collusive behavior must be Byzantine.
However, a stronger notion of equilibria such as \emph{$k$-resilient} equilibria \cite{abraham2011distributed, halpern2008beyond}, which states that no group of $k$ players can be ``better off'' by deviating, could be used to characterize our algorithm's sensitivity to rational collusive behavior. 
In this work, we proved that behaving correctly in Hearsay and RRB constitutes a Nash (i.e., \emph{1-resilient}) equilibrium.
Furthermore, other concepts from economics such as \emph{sequential equilibrium} or \emph{trembling-hand perfect equilibrium}--where agents are allowed to react to observed actions or ``tremble'' and play erroneous strategies with negligible probability, respectively--may be used to further refine the robustness of distributed algorithms to rational behavior.

\paragraph{Optimizations of the Hearsay Protocol}
In the Hearsay protocol specification in Algorithm \ref{algo:protocol}, Condition 2 enforces transaction dependencies which create a partial order beyond simply source order on transaction executions. 
Enforcing dependency ordering is necessary in Hearsay because all agents must be able to agree on whether or not to commit a transaction. 
This ensures that that an agent attempting to spend more money than it has is suitably punished. 
However, additional ordering restrictions cause increased transaction settlement times and inevitably result in system slowdowns. 
A clear optimization possible in Hearsay is to provide an agent initiating a transaction with the option of \textit{which} dependencies to reference in the transaction. 
If an agent already has a balance greater than their transaction amount plus fees, there is no compelling reason to reference any dependencies. 
For simplicity of presentation and analysis we forego this optimization, and we leave it to future work to further characterize the performance tradeoffs of the punishment mechanism used in Hearsay. 
Additionally, we remark that an optimization of this variety loosely resembles the unspent transaction output (UTXO) model used by Bitcoin \cite{Nakamoto2008}, in which each transaction references a set of transaction outputs which were received but not yet spent rather than an account. 

\paragraph{Solving Classical Distributed Computing Problems in Game Theoretic Settings}
We design and make use of RRB in a black-box fashion in Hearsay in part because we believe that RRB may be useful in applications outside of money transfer. 
Distributed peer-to-peer protocols in which the decision-makers in the system are real people rather than computer processes are rapidly growing in popularity.
Applying game theoretic notions of rationality and incentive compatibility to classical distributed computing problems is critical to understanding the interactions between rational agents and maintaining fault tolerance when agents are self-serving. 
Additionally, characterizing the performance impact of the incentive mechanisms used in RRB is an interesting avenue of future work. 

\bibliographystyle{ACM-Reference-Format}
\bibliography{refs}

\pagebreak
\appendix
\section{Rational Reliable Broadcast Algorithm}\label{app:rrbalg}
\begin{algorithm}[htb!]
\scriptsize
\DontPrintSemicolon
\SetKwInOut{Input}{input~}
\SetKwInOut{Local}{Local variables~}
\SetKwInOut{Output}{output~}
\DontPrintSemicolon
\Local{$N$ :: the number of agents in the system \\
$i \in \{1, \dots, N\}$ :: my agent number \\
$s = 1$ :: RRB instance's sequence number\\
$\Vec{l}_e, \Vec{l}_r = 0^N$ :: \emph{echo} and \emph{ready} last sequence received vector\\
$\{Q_\emph{e in}^i\}_{i=1}^{N}$ :: N lists of (sequence number, incoming \emph{echo} message) tuples ordered by sequence number \\
$\{Q_\emph{r in}^i\}_{i=1}^{N}$ :: N lists of (sequence number, incoming \emph{ready} message) tuples ordered by sequence number \\
$\{Q_\emph{out}^i\}_{i=1}^{N}$ :: N lists of (sequence number, outgoing message) tuples ordered by sequence number \\
}
\vspace{\baselineskip}
func \textbf{Broadcast}$(\emph{msg}):$\;
\algotab Send tuple $(\emph{initial}, msg, s)$ to all agents \;
\texttt{\\}

func \textbf{FlushInQueues}$()$:\;
\algotab \For{$1 \leq i \leq N$}
{
\algotab $(\emph{seq}, \emph{msg}) \leftarrow \text{peak}(Q^i_\emph{e in})$\;
\algotab \While{$\Vec{l}_e^i \geq seq - 1$}
    {
    \algotab \lIf{$\Vec{l}_e^i = seq - 1$}{$\Vec{l}_e^i \leftarrow seq$}
    \algotab $\text{pop}(Q^i_\emph{e in})$ until sequence number is greater than $\Vec{l}_e^i$ (so that duplicate messages for the same sequence number are ignored) \label{algoline:ignore1}\;
    \algotab $(\emph{seq}, \emph{msg}) \leftarrow \text{peak}(Q^i_\emph{e in})$\;
    }
 \;  
\algotab $(\emph{seq}, \emph{msg}) \leftarrow \text{peak}(Q^i_\emph{r in})$\;
\algotab \While{$\Vec{l}_r^i = seq - 1$}
    {
    \algotab \lIf{$\Vec{l}_r^i = seq - 1$}{$\Vec{l}_r^i \leftarrow seq$}
    \algotab $\text{pop}(Q^i_\emph{r in})$ until sequence number is greater than $\Vec{l}_r^i$ (so that duplicate messages for the same sequence number are ignored)\label{algoline:ignore2}\;
    \algotab $(\emph{seq}, \emph{msg}) \leftarrow \text{peak}(Q^i_\emph{r in})$\;
    }
}
\texttt{\\}

func \textbf{FlushOutQueue}$()$:\;
\algotab \For{$1 \leq i \leq N$}
{
\algotab $(\emph{seq}, \emph{msg}) \leftarrow \text{peak}(Q^i_\emph{out})$\;
\algotab \While{$\min\{\Vec{l}_e^i, \Vec{l}_r^i\} \nless seq$}
    {
    \algotab send \emph{msg} to $p_i$\;
    \algotab $\text{pop}(Q^i_\emph{out})$\;
    \algotab $(\emph{seq}, \emph{msg}) \leftarrow \text{peak}(Q^i_\emph{out})$\;
    }
}
\texttt{\\}

func \textbf{Deliver}$(msg, s'):$\;
\algotab \lIf{$s \leq s'$}{$s \leftarrow s' + 1$}
\algotab \lIf{A message has not been delivered for sequence $s$}{deliver \emph{msg}}
\texttt{\\}

on \textbf{Receive}$(\emph{initial}, msg, s')$ from $p_j$:\;
\algotab \lIf{\emph{echo} message for sequence $s'$ not already queued}
{push$(Q^j_\emph{out}, (s', (\emph{echo}, msg, s')))$ and \textbf{FlushOutQueue}$()$\;
}

on \textbf{Receive}$(l \in \{\emph{echo}, \emph{ready}\}, msg, s')$ from $p_j$:\;
\algotab \lIf{$l = \emph{echo}$}{push$(Q^j_\emph{e in}, (s', msg))$}
\algotab \lElse{push$(Q^j_\emph{r in}, (s', msg))$}\;

\algotab \textbf{FlushInQueues}$()$\;

\algotab \uIf{$\Big|\big\{i\in\big[N\big] : \Vec{l}_e^i \geq s'\big\}\Big| \geq (n+t)/2\text{ or }\Big|\big\{i\in\big[N\big] : \Vec{l}_r^i \geq s'\big\}\Big| \geq (t+1)$\label{line:echoreadycond}}
{
\algotab \lIf{\emph{echo} message for sequence $s'$ not already queued}
{push$(Q^j_\emph{out}, (s', (\emph{echo}, msg, s')))$\;
}
\algotab \lIf{\emph{ready} message for sequence $s'$ not already queued}
{push$(Q^j_\emph{out}, (s', (\emph{ready}, msg, s')))$\;
}
}
\algotab \lIf{$\Big|\big\{i\in\big[N\big] : \Vec{l}_r^i \geq s'\big\}\Big| \geq (2t+1)$}{$\textbf{Deliver}(msg, s')$}\;
\algotab \textbf{FlushOutQueue}$()$

\caption{Rational Reliable Broadcast}
\label{algo:rrb}
\end{algorithm} 
\newpage
\section{Deferred Proofs}\label{app:defproofs}

\begin{lemma} Hearsay satisfies Validity if no rational agent deviates from the protocol specification.
\end{lemma}
\begin{proof} In this proof, all line numbers refer to lines in the Hearsay Algorithm. 
Additionally, we use the notation $p.S[q]$ to designate the local variable held by agent $p$ which stores the current sequence number of agent $q$.

Let $\mathcal{R}$ be the set of rational agents and $\mathcal{C}$ be the set of past-correct agents.
Let $p$ be an arbitrary agent, then the formal theorem statement is:

\begin{equation}
    (\mathcal{R} \subseteq \mathcal{C} \land \forall p \in \mathcal{R}, p \textbf{ broadcasts } tx) \rightarrow p \textbf{ executes } tx
\end{equation}

To prove this, we first prove the following statement for any agent $p$ which broadcasts a transaction \emph{tx}. 

\begin{equation}
    (\mathcal{R} \subseteq \mathcal{C} \land \forall p \in \mathcal{R},  (p \textbf{ broadcasts } tx \land tx.seq = s \land p.S[p] = s-1 )) \rightarrow p \textbf{ executes } tx \label{eq:validity_helper}
\end{equation}

In the following logic, note that although nearly every logical implication depends it, we omit copying the predicate $\mathcal{R} \subseteq \mathcal{C}$ on each line for the sake of readability. 

\begin{align*}
    &\mathcal{R} \subseteq \mathcal{C} \land \forall p \in \mathcal{R},  (p \textbf{ broadcasts } tx \land tx.seq = s \land p.S[p] = s-1) &&\text{(hypothesis)} \\
    &\implies \forall p \in \mathcal{R},  (p \textbf{ delivers } tx \land tx.seq = s \land p.S[p] = s-1) &&\text{(RRB validity)} \\
    &\implies \forall p \in \mathcal{R},  (p \textbf{ delivers } tx \land (\forall t \in \emph{tx.deps} \cup \emph{tx.fees}, p \textbf{ executes } t) \land tx.seq = s \land p.S[p] = s-1) &&\text{(lines \ref{hearsay:assign_tx}, \ref{hearsay:assign_I}, \ref{hearsay:add_feedependency},
    \ref{hearsay:add_txdependency})} \\
    &\implies \forall p \in \mathcal{R},  (p \textbf{ delivers } tx \land (\forall t \in \emph{tx.deps} \cup \emph{tx.fees}, p \textbf{ executes } t) \\
    &\land tx.seq = s \land p.S[p] = s-1 \land p.B[p] \geq N \epsilon) &&\text{(line \ref{hearsay:check_my_balance})} \\
    &\implies p \textbf{ executes } tx &&\text{(lines \ref{hearsay:buffer}, \ref{hearsay:cond1}, \ref{hearsay:cond2}, \ref{hearsay:cond3})} \\
\end{align*}

Then let $tx_1, \dots, tx_k$ be the set of all transactions broadcast by $p$, in order of increasing sequence number. 
By lines \ref{hearsay:assign_tx} and \ref{hearsay:assign_c}, it is clear that $tx_{k}.seq = k$. 
The inductive argument is then clear. 

Next we prove the base case of the induction argument. 

\begin{equation}
    \mathcal{R} \subseteq \mathcal{C} \land p \in \mathcal{R} \land p \textbf{ broadcasts } tx \land \emph{tx.seq} = 1 \rightarrow p \textbf{ executes } tx
\end{equation}

\begin{align*}
    &\mathcal{R} \subseteq \mathcal{C} \land p \in \mathcal{R} \land p \textbf{ broadcasts } tx \land \emph{tx.seq} = 1 &&\text{(induction hypothesis)} \\
    &\implies p \in \mathcal{R} \land p \textbf{ broadcasts } tx \land \emph{tx.seq} = 1 \land p.S[p] = 0 &&\text{(initial condition)} \\
    &\implies p \textbf{ executes } tx&&\text{(Equation \ref{eq:validity_helper})} \\
\end{align*}

Next, we assume that the following holds for all transaction sequence numbers $\emph{tx.seq} \leq k$. 

\begin{equation}
    \mathcal{R} \subseteq \mathcal{C} \land p \in \mathcal{R} \land p \textbf{ broadcasts } tx \land \emph{tx.seq} \leq k \rightarrow p \textbf{ executes } tx \label{eq:validity_hypothesis}
\end{equation}

And then we prove that the following holds for $\emph{tx.seq} = k + 1$.

\begin{equation}
    \mathcal{R} \subseteq \mathcal{C} \land p \in \mathcal{R} \land p \textbf{ broadcasts } tx \land \emph{tx.seq} = k + 1 \rightarrow p \textbf{ executes } tx
\end{equation}

\begin{align*}
    &\mathcal{R} \subseteq \mathcal{C} \land p \in \mathcal{R} \land p \textbf{ broadcasts } tx \land \emph{tx.seq} = k + 1 &&\text{(hypothesis)}\\
    &\implies p \in \mathcal{R} \land p \textbf{ broadcasts } tx \land \emph{tx.seq} = k + 1 \land p.S[p] = k &&\text{(Equation \ref{eq:validity_hypothesis})}\\
    &\implies p \textbf{ executes } tx &&\text{(Equation \ref{eq:validity_helper})}\\
\end{align*}
\end{proof}
\begin{lemma} Hearsay satisfies Agreement if no rational agent deviates from the protocol specification.
\end{lemma}
\begin{proof}The formal statement is

\begin{equation}
    (\mathcal{R} \subseteq \mathcal{C} \land \exists p \in \mathcal{R}: p \textbf{ executes } tx) \rightarrow (\forall q \in \mathcal{R}, q \textbf{ executes } tx \land (p \textbf{ commits } tx \rightarrow q \textbf{ commits } tx))
\end{equation}

For simplicity of notation and understanding, the predicates and implications used in this proof need not hold at all points in time, as the network is asynchronous and events are only partially ordered.
Instead, the logical implications in this proof must \textit{eventually} hold. 
For example, if all agents except for one have executed a transaction, it may be the case that the final agent has not yet executed the transaction but is guaranteed to do so in the future.
Additionally, when equality is stated between local variables from two different agents, it means only that at the moment the variables are accessed by each agent in the context of the surrounding logic, they are guaranteed to contain the same value. 
Also note that although nearly every logical implication depends it, we omit copying the predicate $\mathcal{R} \subseteq \mathcal{C}$ on each line for the sake of readability. 

We define the following sets and quantities to be used in the proof below. 

Let $p$ be some rational agent which executes a transaction $tx$.
Assume that the total number of transactions initiated is finite.
Let $\Phi(tx)$ be a set of transactions defined recursively as follows:

\begin{equation}
    t \in \Phi(tx) \iff (t=tx) \lor (\exists v \in \Phi(tx): (t \in \emph{v.deps}) \lor (t \in  \emph{v.fees}) \lor (\emph{t.initiator} = \emph{v.initiator} \land \emph{t.seq} \leq \emph{v.seq}) )
\end{equation}

In other words, $\Phi(tx)$ is the set of all transactions that a correct agent must execute prior to executing $tx$. 
It is clear that for any $tx$,

\begin{equation}
    \forall t \in \Phi(tx) \setminus tx, \Phi(t) \subsetneq \Phi(tx) \label{eq:d_subset}
\end{equation}

Additionally, let $h(tx) \in \mathbb{N}$ be the \emph{height} of $tx$, where height is defined as 

\begin{equation}
    h(tx) = \begin{cases} 
    1\;&\text{if $\Phi (tx) = \{tx\}$}\\
    1 + \underset{t \in \Phi(tx) \setminus tx}{\max} h(t),\;&\text{otherwise.} \end{cases}\label{eq:height}
\end{equation}

Intuitively, the height of $tx$ is the length of the longest chain of transaction which a correct agent must execute prior to $tx$. 

Let $p.\Delta(tx) \in \mathbb{R}^N$ be the change in agent $p$'s local view of the balances of each other agent which is caused by $p$ executing $tx$. 
It is clear from the \textbf{Pay} function that a correct agent $p$ will only modify the balance of \emph{tx.initiator} for any transaction $tx$ that $p$ executes. More specifically,

\begin{equation}
    \forall j : (j \neq \emph{tx.initiator}), p.\Delta(tx)[j] = 0 \label{eq:only_initiator}
\end{equation}

We use the following helper statement. 

\begin{equation}
\begin{split}
    &\mathcal{R} \subseteq \mathcal{C} \land \big((\exists p \in \mathcal{R}, \forall q \in \mathcal{R}, \forall t \in \Phi(tx) \setminus tx, (p \textbf{ executes } tx) \land (q \textbf{ executes } t) \land (p \textbf{ commits } t \rightarrow q \textbf{ commits } t)) \\
    &\rightarrow (q \textbf{ executes } tx \land (p \textbf{ commits } tx \rightarrow q \textbf{ commits } tx))\big)
\end{split}
\label{eq:agreement_helper}
\end{equation}

This helper statement states that if some rational agent $p$ executes $tx$, all rational agents execute all dependencies of $tx$, and all rational agents agree on whether or not to commit each dependency of $tx$, then all rational agents execute $tx$ and agree on whether or not to commit $tx$.
We prove the helper statement with the following logic. 

\begin{align*}
    \mathcal{R} \subseteq \mathcal{C} \land &(\exists p \in \mathcal{R}, \forall q \in \mathcal{R}, \forall t \in \Phi(tx) \setminus tx, (p \textbf{ executes } tx) \land (q \textbf{ executes } t) \land (p \textbf{ commits } t \rightarrow q \textbf{ commits } t)) &&\text{(hypothesis)}\\
    \implies &\exists p \in \mathcal{R}, \forall q \in \mathcal{R}, \forall t \in \Phi(tx) \setminus tx, (p \textbf{ delivers } tx) \land (q \textbf{ executes } t) \land (p \textbf{ commits } t \rightarrow q \textbf{ commits } t) &&\text{(line \ref{hearsay:buffer})}\\
    \implies &\exists p \in \mathcal{R}, \forall q \in \mathcal{R}, \forall t \in \Phi(tx) \setminus tx, (q \textbf{ delivers } tx) \land (q \textbf{ executes } t) \land (p \textbf{ commits } t \rightarrow q \textbf{ commits } t) &&\text{(RRB Agreement)}\\
    \implies &\exists p \in \mathcal{R}, \forall q \in \mathcal{R}, \forall t \in \Phi(tx) \setminus tx, (q \textbf{ delivers } tx) \land (q \textbf{ executes } t) \land (p \textbf{ commits } t \rightarrow q \textbf{ commits } t) &&\text{(lines \ref{hearsay:add_payments}, \ref{hearsay:add_fees}, \ref{hearsay:subtract_fees}, \ref{hearsay:subtract_payment};}\\
    &\land (p.\Delta(t) = q.\Delta(t))  &&\text{Equation \ref{eq:only_initiator})}\\
    \implies &\exists p \in \mathcal{R}, \forall q \in \mathcal{R}, \forall t \in \Phi(tx) \setminus tx, (q \textbf{ delivers } tx) \land (q \textbf{ executes } t) \land (p \textbf{ commits } t \rightarrow q \textbf{ commits } t) \\
    &\land (p.B[\emph{tx.initiator}] = q.B[\emph{tx.initiator}]) &&\text{(definition of $\Delta(tx)$)}\\
    \implies &q \textbf{ executes } tx \land (p \textbf{ commits } tx \rightarrow q \textbf{ commits } tx)&&\text{(lines \ref{hearsay:cond1}, \ref{hearsay:cond2},
    \ref{hearsay:cond3}, \ref{hearsay:start}--\ref{hearsay:bad_assign})}
\end{align*}

Next we state the base case of our induction argument. 

\begin{equation}
    (\mathcal{R} \subseteq \mathcal{C} \land \exists p \in \mathcal{R}: p \textbf{ executes } tx \land h(tx) = 1) \rightarrow (\forall q \in \mathcal{R}, q \textbf{ executes } tx \land (p \textbf{ commits } tx \rightarrow q \textbf{ commits } tx))
\end{equation}

The proof of the base case uses the following logic. 

\begin{align*}
    \mathcal{R} \subseteq \mathcal{C} \land &\exists p \in \mathcal{R}: p \textbf{ executes } tx \land h(tx) = 1 &&\text{(hypothesis)}\\
    \implies &\exists p \in \mathcal{R}: p \textbf{ executes } tx \land (\Phi(tx) \setminus tx = \varnothing) &&\text{(Definition of $h(tx)$)}\\
    \implies &\forall q \in \mathcal{R}, q \textbf{ executes } tx \land (p \textbf{ commits } tx \rightarrow q \textbf{ commits } tx)&&\text{(Equation \ref{eq:agreement_helper})}
\end{align*}

To complete the argument, we assume that the following statement holds for all $tx$ and for any $h(tx) \leq k$:

\begin{equation}
    (\mathcal{R} \subseteq \mathcal{C} \land \exists p \in \mathcal{R}: p \textbf{ executes } tx \land h(tx) \leq k) \rightarrow (\forall q \in \mathcal{R}, q \textbf{ executes } tx \land (p \textbf{ commits } tx \rightarrow q \textbf{ commits } tx)) \label{eq:agreement_hypothesis}
\end{equation}

And then we perform induction over $h(tx)$: 

\begin{equation}
    (\mathcal{R} \subseteq \mathcal{C} \land \exists p \in \mathcal{R}: p \textbf{ executes } tx \land h(tx) = k+1) \rightarrow (\forall q \in \mathcal{R}, q \textbf{ executes } tx \land (p \textbf{ commits } tx \rightarrow q \textbf{ commits } tx))
\end{equation}

The proof of the inductive argument uses the following logic.

\begin{align*}
    \mathcal{R} \subseteq \mathcal{C} \land &\exists p \in \mathcal{R}: p \textbf{ executes } tx \land h(tx) = k+1 &&\text{(hypothesis)} \\
    \implies &\exists p \in \mathcal{R}: p \textbf{ executes } tx \land (\forall q \in \mathcal{R}, \forall t \in \Phi(tx) \setminus tx : h(t) \leq k, \\
    &q \textbf{ executes } t \land (p \textbf{ commits } t \rightarrow q \textbf{ commits } t)) &&\text{(Equation \ref{eq:agreement_hypothesis})} \\
    \implies &q \textbf{ executes } tx \land (p \textbf{ commits } tx \rightarrow q \textbf{ commits } tx)&&\text{(Equation \ref{eq:agreement_helper})}
\end{align*}
\end{proof}

\begin{lemma} Hearsay satisfies Integrity if no rational agent deviates from the protocol specification.
\end{lemma}
\begin{proof}[Formal Proof]
Let $tx_1$ and $tx_2$ be any two transactions. 

Formally, the lemma states:

\begin{equation}
\begin{split}
    &\mathcal{R} \subseteq \mathcal{C} \land (\forall p \in \mathcal{R}, p \textbf{ executes } tx_1 \land p \textbf{ executes } tx_2) \\
    &\rightarrow (tx_1.initiator \neq tx_2.initiator \lor tx_1.seq \neq tx_2.seq)
\end{split}
\end{equation}

We prove this with the following logic. 
Let $p$ be an arbitrary agent and let $tx_1, tx_2$ be two arbitrary transactions. 

\begin{align*}
    &\mathcal{R} \subseteq \mathcal{C} \land p \in \mathcal{R} \land p \textbf{ executes } tx_1 \land p \textbf{ executes } tx_2 &&\text{(hypothesis)} \\
    &\implies p \in \mathcal{R} \land p \textbf{ delivers } tx_1 \land p \textbf{ delivers } tx_2 &&\text{(line \ref{hearsay:buffer})} \\
    &\implies tx_1.initiator \neq tx_2.initiator \lor tx_1.seq \neq tx_2.seq &&\text{(RRB Integrity)}
\end{align*}
\end{proof}

\end{document}